\documentclass[journal]{IEEEtran}

\usepackage{subfig}
\usepackage{amssymb,amsmath}
\usepackage{graphicx}
\usepackage{amsmath}

\usepackage{color}
\usepackage{algorithmic}
\usepackage{algorithm}
\newtheorem{theorem}{Theorem}[section]
    \newtheorem{lemma}[theorem]{Lemma}
    \newtheorem{example}{Example}
    \newtheorem{definition}[theorem]{Definition}

    \newtheorem{proof}{Proof}
    
\input xy
\xyoption{all}

\begin{document}

\title{On the Vulnerability of Digital Fingerprinting Systems to Finite Alphabet Collusion Attacks}

\author{Jalal~Etesami,~\IEEEmembership{Student Member,~IEEE,}
        and~Negar~Kiyavash~\IEEEmembership{Senior Member,~IEEE}
\thanks{J. Etesami is with the Department of Industrial and Enterprising Systems Engineering and the Coordinated Science Laboratory, University of Illinois at Champaign-Urbana, Urbana, Illinois 6180, USA. 
(e-mail: etesami2@illinois.edu).}
\thanks{N. Kiyavash is with the Department of Industrial and Enterprising Systems Engineering and the Coordinated Science Laboratory, University of Illinois at Champaign-Urbana, Urbana, Illinois 6180, USA. 
(e-mail: kiyavash@illinois.edu).}
}

\maketitle

\begin{abstract}
This paper proposes a novel, non-linear collusion attack on digital fingerprinting systems. The attack is proposed for fingerprinting systems with finite alphabet but can be extended to continuous alphabet.  We analyze the error probability of the attack for some classes of proposed random and deterministic schemes and obtain a bound on the number of colluders necessary to correctly estimate the host signal. That is, it requires fewer number of colluders to defeat the fingerprinting scheme. Our simulation results show that our attack is more powerful in practice than predicted by the theoretical bound.
\end{abstract}

\section{Introduction}
\IEEEPARstart{D}{igital} fingerprinting schemes have been invented for traitor tracing as a means of copyright protection. To this end, each user is provided with his own individually modified copy of the host signal (i.e., the content that is copyright protected). The modification entails embedding a unique signature, also known as a fingerprint, in the host signal. The fingerprint can later be retrieved by a detector examining an illegal copy of the content in order to implicate the users who took part in the forgery.

A myriad of fingerprint designs and detection procedures have been proposed in the literature. Examples include orthogonal fingerprints \cite{wang}, regular simplex fingerprints \cite{n2}, random Gaussian fingerprints \cite{n1}, \cite{gauss}, and fingerprints with equiangular tight frames (ETF) \cite{ETF}. The main consideration in design of fingerprints is their rate, i.e., number of copies that could be distributed and their robustness against attacks, in particular collusion attacks.

In a collusion attack, a group of users combine their copies to create a forgery in order to defeat the detection process. For instance, one such type of collusion attack is linear averaging, in which the colluders estimate the host signal as the linear average of their marked copies. This attack is studied in numerous works  \cite{n3,n4,nn3}. 

A novel collusion attack is introduced in \cite{jalalnk} that proves detrimental against both random and deterministic fingerprints as long as they are from a finite alphabet.  The entries of random fingerprints are chosen i.i.d. from a finite set according to a distribution while deterministic fingerprints are designed non-randomly often according to some geometric or algebraic property. 
The analysis in \cite{jalalnk} shows that if the size of collusion is of order $O(\log N)$ where $N$ is the length of a real-valued host signal, the attackers can accurately estimate the host signal with high probability.  In particular, this attack improves over the previous best-known attack against ETF fingerprints \cite{ETF, design}, i.e., the linear averaging attack.  It was shown that finite alphabet ETF random or deterministic fingerprints can withstand linear averaging attack as long as number of colluders is bounded by $O(N/\log^{a}M)$ and $O(\sqrt{N})$ respectively, where $M$ denotes the total number of fingerprinted copies \cite{ETF}.

In this work, we analyze the error probability of the aforementioned collusion attack against uniformly symmetric random fingerprints and obtain a bound on the number of colluders necessary to succeed with high probability. 
Furthermore, we study the performance of the attack against two more probabilistic fingerprints: optimal probabilistic fingerprint (OPF) codes, aka Tardos codes, and column-wise random fingerprints (Tardos-like codes).
The OPF codes are introduced in \cite{tardos} that are $\epsilon$-secure against $K$ pirates and have length $O(K^2 \log 1/\epsilon)$. Our analysis shows that our collusion attack outperforms the linear averaging \cite{minority}, minority voting \cite{minority}, and the majority voting attacks \cite{minority2}, the latter two known to be effective collusion attacks against the OPF codes.  

We further generalize the application of this collusion attack to continuous alphabet codes, specifically Gaussian fingerprints by quantizing their support set. 
Our analysis shows a good performance of the generalized attack compared with the uniform linear averaging attack, proved to be optimal among the class of order-statistic collusion attacks, when the colluders are subject to a mean squared distortion constraint \cite{nn3}.

The rest of paper is organized as follows: Section \ref{sec:set} describes the mathematical model of the fingerprinting problem. In Section \ref{att}, we present the proposed collusion attack. Section \ref{sec:etf}, \ref{estim}, and \ref{sec:pro} examine the attack's effectiveness on both deterministic and random fingerprints and propose lower bounds on the number of colluders that can break the fingerprinting system with high probability. In Section \ref{simulation}, we use simulations to demonstrate the effectiveness of our attack on two types of fingerprint codes. We present our conclusions in Section \ref{con}.

Throughout this paper, we use bold upper-case letters to denote matrices, boldface for vectors, lower-case for scalar values, and calligraphic fonts for sets.

\section{Problem Setup}\label{sec:set}
In this section we describe our mathematical setup for fingerprint embedding. 

Consider a host signal $\textbf{s}\in \mathbb{R}^{N}$ and $M(> N)$ marked copies of the host signal $\textbf{s}$ distributed among $M$ users. In particular, the $m$th user receives the copy,
\begin{equation}\label{1}
    \textbf{q}_{m}:=\textbf{f}_{m}+\textbf{s} \;,      
\end{equation}
where $\textbf{f}_{m}\in \mathbb{R}^{N}$ is the $m$th fingerprint. We assume that all fingerprints have equal energy, i.e.,
\begin{equation}\label{denergy}
||\textbf{f}_{m}||^{2}_{2}=1\ , m=1,...,M,  
\end{equation}
in the case of deterministic construction, and
\begin{equation}\label{energy}
\mathbb{E}[||\textbf{f}_{m}||^{2}]=1\ , m=1,...,M,
\end{equation}
in the case of random construction. 

Let $\textbf{F}$ denote the $N\times M$ matrix whose columns are the fingerprints $\{\textbf{f}_{m}\}_{m=1}^{M}$, and let $\{\textbf{e}_{m}\}_{m=1}^{M}$ denote the standard basis in $\mathbb{R}^{M}$, i.e., for any $m\in\{1,...,M\}$, $e_{m}(i)=0$ when  $i\neq m$ and 1 otherwise. Then (\ref{1}) may be written as
\begin{equation}\label{2}
    \textbf{q}_{m}=\textbf{F}\textbf{e}_{m}+\textbf{s} \;.
\end{equation}

The goal of the fingerprinting process is to prevent users from illegally sharing their copies. Therefore, the fingerprint detector must be able to detect the users or some subset of the users who took part in the forgery.

\section{Attack Strategy}\label{att}
In this section we introduce a new type of collusion attack against fingerprints from a finite alphabet.

 Consider a fingerprinting matrix $\textbf{F}=[f_{i,j}]_{N\times M}$ where the entries are chosen deterministically or at random from a finite set $\Xi=\{\xi_{1},...,\xi_{l}\}\subset\mathbb{R}$, such that $\xi_{1}<\xi_{2}<...<\xi_{l}$. Given that the attackers know $\Xi=\{\xi_{1},...,\xi_{l}\}$, they can construct the following set:
\begin{equation}\label{du}
\mathcal{U}=\{u\in\mathbb{R} :\ \exists\ \textsl{unique}\ (i,j) \ s.t.\ u=\xi_{i}-\xi_{j} \}.
\end{equation}
Each element of set $\mathcal{U}$ is a real number that corresponds to the difference of a unique ordered pair in $\Xi\times \Xi$. Note that $|\mathcal{U}|\geq2$ because $\pm(\xi_{l}-\xi_{1})$ are always in $\mathcal{U}$.

Suppose $K$ users collude with the goal of learning at least one of their fingerprints and subsequently obtaining the host signal. Without loss of generality denote their copies by $\{\textbf{q}_{1},...,\textbf{q}_{K}\}$. The attackers form the $N\times K$ matrix $\textbf{Q}=[\textbf{q}_{1} ... \textbf{q}_{K}]$. Furthermore, they choose one of the copies, e.g., $\textbf{q}_{1}$ (it will become clear that this choice is immaterial) and compute the $N\times K$ matrix $\textbf{A}$ whose $i$th column is given by $\textbf{q}_{1}-\textbf{q}_{i}$ or equivalently $\textbf{f}_{1}-\textbf{f}_{i}$. Denote the $i$th row of $\textbf{A}$ by $\textbf{a}_{i}$.  Whenever the attackers encounter an element of $\textbf{A}$, $a_{i,j}$ that belongs to $\mathcal{U}$, they can uniquely determine $(\xi,\xi')\in\Xi\times\Xi$ s.t. $a_{i,j}=\xi-\xi'$. To illustrate how this helps the attackers learn one of the fingerprints, we provide the following example:

\begin{example}
\textit{Assume 4 colluders took part in the forgery. Furthermore, assume $\Xi=\{-1,0,1\}$ and let the $i$th row of $\textbf{Q}$ be $[s_{i}, s_{i}+1,s_{i}+1,s_{i}-1]$. Hence, $\textbf{a}_{i}=[0,-1,-1,1]$ and $\mathcal{U}=\{-2,2\}$. Although no element of $\textbf{a}_{i}$ is in $\mathcal{U}$, by computing pairwise differences of elements of $\textbf{a}_{i}$, the colluders can learn $f_{i,1}$ as follows: $a_{i,2}-a_{i,4}=(f_{i,1}-f_{i,2})-(f_{i,1}-f_{i,4})=f_{i,4}-f_{i,2}=-2$. Therefore, $f_{i,4}=-1, f_{i,2}=1$ and subsequently from $f_{i,1}-f_{i,4}=1$, it follows that $f_{i,1}=0$.  }
 \end{example}

As illustrated by Example 1, looking at entries of $\textbf{A}$ alone was not enough to learn the fingerprints. Therefore, for each row $i$ of matrix $\textbf{A}$ the attackers compute the set:
\begin{equation}\label{ai}
\mathcal{A}_{i}:=\{\Delta a_{i}: \Delta a_{i}=a_{i,j}-a_{i,k}:\ k,j=1,...,K\}.
\end{equation}
If  $\Delta a_{i}\in \mathcal{U}$, then attackers can learn $f_{i,1}$. This is because $a_{i,j}=f_{i,1}-f_{i,j}$ and $a_{i,k}=f_{i,1}-f_{i,k}$. Thus $\Delta a_{i}=f_{i,k}-f_{i,j}$ and following the definition of set $\mathcal{U}$, $(f_{i,k},f_{i,j})$ is uniquely determined. Subsequently $f_{i,1}$ can be computed as
\begin{equation}\label{es1}
f_{i,1}=a_{i,j}+f_{i,j},
\end{equation}
Once $f_{i,1}$ is known, all others $f_{i,j}$, can be determined from (\ref{es1}).
If no elements of $\mathcal{A}_{i}$: $\Delta a_{i}\notin \mathcal{U}$, the attackers estimate $f_{i,1}$ by the first component of  $\textbf{b}^{*}_{i}=(b^{*}_{i,1},...,b^{*}_{i,K})$, where
\begin{equation}\label{otherwise}
\textbf{b}_{i}^{*}=\arg\max_{\textbf{b}\in\mathcal{B}(\textbf{a}_{i})}p(\textbf{b}),
\end{equation}
where $p(\textbf{b})$ denotes the probability of vector $\textbf{b}$ and $\mathcal{B}(\textbf{a}_{i})$ is given by
\begin{equation}\label{Bi}
\mathcal{B}(\textbf{a}_{i})\small{:=\left\{\textbf{b}=(b_{1},...,b_{K})\in\Xi^{K}:\ b_{1}-b_{j}=a_{i,j},\ j=1,...,K\right\}}.
\end{equation}
That is, the attackers choose the most likely fingerprints that are consistent with the set of observations $\mathcal{A}_i$.
 If there are several solutions to (\ref{otherwise}), the attackers choose one of them arbitrarily. The attackers compute $p(\textbf{b})$ using their side information about the structure of fingerprints. For instance, if they know that the entries of $\textbf{F}$ were chosen i.i.d. at random according  to $\textbf{p}(\Xi)=(p_{\xi_{1}},...,p_{\xi_{l}})$, then
\begin{equation}\label{rnd}
p(\textbf{b})=\prod_{i=1}^{K}p_{b_{i}}.
\end{equation}
Recall that attackers could learn $f_{i,1}$ if they could use the Equation (\ref{es1}). Otherwise, their failure probability depends on how $\textbf{F}$ was generated. Note that by minimizing $|\mathcal{U}|$, the designer increases the likelihood of attackers' failure regardless of the structure of $\textbf{F}$.

\section{Attack on Equiangular Tight Frame}\label{sec:etf}
In this section we study the performance of our attack on a deterministic fingerprint known as equiangular tight frame (ETF) fingerprints from a finite alphabet, proposed in \cite{design}.
Although it was shown in \cite{ETF} that deterministic ETF fingerprints with finite alphabet are robust against the linear averaging attack, namely, if the size of collusion is $O(\sqrt{N})$, then with high probability the focused detector, which performs a binary hypothesis test for each user to decide whether that particular user is guilty \cite{ETF}, is able to detect at least one of the attackers,  we will show that these type of fingerprints can not withstand the proposed attack.

 A \textit{frame} is a collection of vectors $\{\textbf{f}_{m}\}_{1}^{M}\in\mathbb{R}^{N}$ with frame bounds $0 < A \leq B < \infty$ such that
$$
A||\textbf{x}||^{2}\leq \sum_{m=1}^{M}|<\textbf{x},\textbf{f}_{m}>|^{2}\leq B||\textbf{x}||^{2},
$$
for every $\textbf{x}\in\mathbb{R}^{N}$. The frame is tight if $A=B$. An equiangular tight frame (ETF) is a unit norm tight frame with the additional property that there exists a constant $c$ such that
$$
|<\textbf{f}_{m'},\textbf{f}_{m}>|=c, \ \forall\ m\neq m'.
$$
The construction in \cite{design} uses a tensor-like combination of a Steiner system's adjacency matrix and a regular simplex.

A Steiner system \cite{steiner} with parameters $r, h, n$, written as $S(r,h,n)$, is an $n$-element set $\mathcal{S}$ together with a set of $h$-element subsets of $\mathcal{S}$ (called blocks) with the property that each $r$-element subset of $\mathcal{S}$ is contained in exactly one block. The matrix representation of this system is a ${n\choose r}/{h \choose r}\times n$ binary matrix such that each row contains precisely $h$ ones and each column has ${n-1\choose r}/{h-1 \choose r}$ ones. Let $H_{m_{0}\times m_{0}}$ to be a Hadamard matrix of order $m_{0}$ \cite{hadamard}. The proposed ETF design in \cite{design} is obtained by substituting non-zero elements of a $S(r,h,n)$ with an arbitrary row of $H$ which contains $-1$. Therefore, the fingerprint matrix $\textbf{F}_{ETF}(r,h,n,m_{0})$ is an $N\times M$ matrix with entries in $\{-1,0,+1\}$ where $M=m_{0}n$ and $N={n\choose r}/{h \choose r}$. It is  usually assumed that $n>8h$.

Next result derives a lower bound on the number of attackers that can break the fingerprinting system given by $\textbf{F}_{ETF}(r,h,n,m_{0})$. The system is broken if the attackers can estimate the host signal $\textbf{s}$ accurately with high probability.

\begin{theorem}\label{attacktheorem}
\textit{Let $\textbf{F}_{ETF}(r,h,n,m_{0})$ denote the ETF fingerprinting matrix and let  $K$ denote the size of coalition set. If
\begin{equation}\label{majority}
K\geq2\log4N/\delta,
\end{equation}
where $N={n\choose r}/{h \choose r}$, then the host signal can be correctly estimated with probability greater than $1-\delta$.}
\end{theorem}
\begin{proof}
Proof is in Appendix.
\end{proof}

\section{Attack on Random Structure}\label{estim}
In this section, we study the performance of the proposed attack on random fingerprints from a finite alphabet. 

Random finite fingerprint matrix $\textbf{F}$ is an $N$ by $M$ matrix whose entries are chosen i.i.d at random from $\Xi=\{\xi_{1},...,\xi_{l}\}$ with probability $p(f_{i,j}=\xi_{k})=p_{\xi_{k}}$ for $k\in\{1,...,l\}$. Moreover, we assume that the fingerprints satisfy the energy constraint (\ref{energy}) and without loss of generality, we assume $\mathbb{E}[f_{i,j}]=0.$

Recall that the attackers might fail to estimate the $i$th component of the host signal correctly if $\mathcal{A}_{i}\cap \mathcal{U}=\emptyset$, where $\mathcal{A}_{i}$ and $\mathcal{U}$ are given by (\ref{ai}) and (\ref{du}), respectively. Let $\textbf{b}\in\Xi^{K}$ and $\tilde{b}_{i}:=b_1-b_i$; then $\mathcal{B}(\tilde{\textbf{b}})$, where $\mathcal{B}(\cdot)$ is given by (\ref{Bi}), denotes the set of all vectors in $\Xi^{K}$ that are consistent with the observation vector $\tilde{\textbf{b}}$.

Attackers estimate $\textbf{b}$ by $\textbf{c}^{*}(\tilde{\textbf{b}})=\arg\max_{\textbf{c}\in\mathcal{B}(\tilde{\textbf{b}})}p(\textbf{c})$ that is the most likely vector in $\mathcal{B}(\tilde{\textbf{b}})$. Since the entries of $\textbf{F}$ are i.i.d., the probability of error in estimating any component of $\textbf{s}$ is the same and is given by
\begin{equation}\label{Error}
p(E)=\sum_{\textbf{b}\in\Xi^{K}}p(\textbf{b})p\left(\textbf{b}\neq \textbf{c}^{*}(\tilde{\textbf{b}})\right),
\end{equation}
where $p(\textbf{b})$ is defined in (\ref{rnd}).

Computing (\ref{Error}) in a general setting might be complicated. However, there are some settings in which we can compute $p(E)$ explicitly or at least bound it. The rest of this section focuses on such scenarios.
\begin{definition}
A random fingerprinting matrix $\textbf{F}_{N\times M}$ with parameters $w\in\mathbb{N}$ and $\textbf{p}=(p_{0},...,p_{w})>0$ is called symmetric if $f_{i,j}$ are chosen i.i.d. from $\Xi=\{-\frac{w}{z},...,-\frac{1}{z},0,\frac{1}{z},...,\frac{w}{z}\}$, such that $p(f_{i,j}=\frac{k}{z})=p(f_{i,j}=-\frac{k}{z})=p_{k}$ for every $k\in\{0,1,...,w\}$.

Note that $z$ is the normalization factor for the fingerprints in order to satisfy (\ref{energy}). Furthermore, $F_{N\times M}$ is called \textit{uniformly symmetric} if $p_{k}=\dfrac{1}{2w+1}$ for all $k$.
\end{definition}

\begin{lemma}\label{errorpro}
\textit{The proposed attack succeeds against a uniformly symmetric fingerprinting system with probability greater than $1-\delta$ as long as the number of colluders, $K$ satisfies
$$
K\geq \frac{\log(N/\delta)}{\log(1+\frac{1}{2w})}.
$$
}
\end{lemma}
\begin{proof}
Proof is in Appendix.
\end{proof}

In general, the source distribution $\textbf{p}_{\Xi}=(p_{1},....,p_{l})$ are not known to the attackers. However, in this case, attackers can estimate the source distribution, for example, using an empirical estimator over those rows of $\textbf{A}$ that have non-empty intersection with the set $\mathcal{U}$. Namely, they find $\mathcal{I}:=\{i: \mathcal{A}_{i}\cap\mathcal{U}\neq\emptyset\}$. As discussed in previous section, the attackers can determine $(f_{i,1},...,f_{i,K})$ correctly for every $i\in\mathcal{I}$. Because the entries of $\textbf{F}$ are chosen i.i.d., the attackers can estimate $\textbf{p}_{\Xi}$ over the set $\{f_{i,j}: i\in\mathcal{I}, 1\leq j\leq K\}$.

\subsection{Random Ternary Fingerprints}\label{RTF}

A random ternary fingerprint (RTF) is an $N\times M$ random symmetric fingerprinting system with parameter $w=1$, and $\textbf{p}=(1-2p,p)$. Next result gives a lower bound on the number of attackers that can detect the host signal with high probability.

\begin{lemma}\label{lemma22}
\textit{An application of the proposed attack on an RTF system succeeds with probability greater than $1-\delta$ as long as
\begin{equation}\label{hmmmm}
K\geq \frac{\log(\frac{2N}{\delta})}{\log4/3},
\end{equation}
where $K$ is the number of colluders. }
\end{lemma}
\begin{proof}
Proof is in Appendix.
\end{proof}

In Section \ref{simulation} we will see the performance of the proposed attack on $RTF$s for different values of $N$ and $M$ through simulations.

\section{Attack on Probabilistic Fingerprint Codes}\label{sec:pro}

This section studies the performance of the proposed attack on binary random fingerprinting codes for $M$ users with length $N$. Specifically, we consider two such fingerprinting structures. In the first structure which is proposed in \cite{tardos}, each row of the fingerprinting matrix $\textbf{F}$ is generated independently from a fixed distribution. In the second structure, each user (column), independent of other users, generates its own fingerprints.

\subsection{Row-Wise Generated (Tardos Code)}\label{tard}

In this section we perform our attack on a probabilistic structure known as optimal probabilistic fingerprint  (OPF) Codes, introduced in \cite{tardos}. These codes are an improvement of the codes proposed by Boneh-Shaw in \cite{boneh_shaw}. The OPF codes are $\epsilon$-secure (see the following definition) against $K$ pirates and have length $O(K^{2}\log{1/\epsilon})$. 

\begin{definition}\cite{tardos}
Let $\sigma$ be an algorithm  that takes a string $\textbf{y}\in\Xi^{N}$ (forged copy by the colluders) as input, and produces a subset $\sigma(\textbf{y})\subseteq\{1,...,M\}$ (the set of accused users). For $\emptyset\neq \mathcal{K}\subseteq\{1,...,M\}$ a $\mathcal{K}$-strategy is an algorithm $\rho$ that takes the sub-matrix of $\textbf{F}_{N\times M}$ formed by the rows with indices in $\mathcal{K}$ as input, and generate a forgery $\textbf{y}\in\Xi^{N}$ as output and satisfies the \textit{marking condition} that for all positions $1\leq i\leq N$, if all the values $F_{i,j}$ for $j\in \mathcal{K}$ agree with some letter $\xi\in \Xi$, then $y_{i}=\xi$. A code is called $\epsilon$-secure against coalition of size $K$, if for any $\mathcal{K}\subseteq \{1,...,M\}$ with $|\mathcal{K}|\leq K$ and for any $\mathcal{K}$-strategy $\rho$, the error probability 
$$
p\left(\sigma(\rho(\textbf{F}_{\mathcal{K}}))=\emptyset \ or\ \sigma(\rho(\textbf{F}_{\mathcal{K}}))\nsubseteq \mathcal{K}\right)\leq \epsilon.
$$
\end{definition}
The code is constructed as follows \cite{tardos}: Let $c=\lceil\log(1/\epsilon)\rceil$. The fingerprinting matrix $\textbf{F}_{OPF}$ is of size $(100K^{2}c=)N\times M$. Let $\varrho_{i}$ be independent, identically distributed random variable from $[1/(300K),1-1/(300K)]$ for all $i\in\{1,...,N\}$, where $\varrho_{i}=\sin^{2}(r_{i})$ and $r_{i}\sim \text{Uniform}[t,\pi/2-t]$ with $0<t<\pi/4$, $\sin^{2}(t)=1/(300K)$. \\
$\textbf{F}_{OPF}$ is constructed by selecting entry $f_{i,j}$ independently from $\{0,1\}$ with $p(f_{i,j}=1)=\varrho_{i}$ for $i\in\{1,...,N\}$.

Notice that the OPF code is similar to the previous structure, except that, the fingerprinting matrix is constructed row-wise (row $i$ with probability $\varrho_{i}$) instead of column-wise (column $j$ with probability $p_{j}$).

The detection procedure for this structure is as follows \cite{tardos}. The detector defines matrix $\textbf{U}_{N\times M}$ such that 
$$
u_{i,j} =
\begin{cases}
\sqrt{\frac{1-\varrho_{i}}{\varrho_{i}}}, & \text{if } f_{i,j}=1\\
-\sqrt{\frac{\varrho_{i}}{1-\varrho_{i}}}, & \text{if } f_{i,j}=0
\end{cases}
$$
and $Z:=20cK$. Then it accuses user $j$ whenever $(\textbf{y}^{T}\textbf{U})_{j}>Z$, where $\textbf{y}$ is the forgery.

Similar to the previous section, consider the coalition $\mathcal{K}=\{1,...,K\}$ who performs the proposed attack in Section \ref{att} on this structure. The colluders estimate the probabilities by
\begin{equation}
\widehat{\varrho}_{i}:=\frac{m_{i}}{|\mathcal{K}|}, \ i\in\mathcal{I},
\end{equation}
where $m_{i}:=|\{f_{i,j}=1:\ j\in\mathcal{K}\}|$. The attacker's estimation of $\textbf{f}_{1}$ will be 
$$
\widehat{f}_{i,1}=\begin{cases}
f_{i,1} & \text{if } i\in\mathcal{I},\\
1 & \text{otherwise}.
\end{cases}
$$
\begin{theorem}\label{tardos2}
Suppose a coalition of size $K\geq4$ perform the proposed attack on the above fingerprinting system with matrix $\textbf{F}_{N\times M}$. Then the number of rows that they fail to estimate correctly ($N-|\mathcal{I}|$) is bounded by $2CN/\sqrt{K}$ with probability at least 
$$
1-\frac{\sqrt{K}-C}{NC},
$$
where $1<C<1.3$ is a constant.
\end{theorem}
\begin{proof}
Proof is in Appendix.
\end{proof}

Colluders use Algorithm \ref{alg} to forge $\textbf{y}$ and hide their tracks as follows: They estimate the signal $\hat{\textbf{s}}$ and its Fourier transform $\mathcal{F}[\widehat{\textbf{s}}]$. Then add  Gaussian and uniform noises to the magnitude and the  phase, respectively. They add  a random vector $c\textbf{w}$ in the time domain to form the final forgery. Vector $\textbf{w}$ is constructed by setting the elements at the indices that the attackers can detect (i.e., $\mathcal{I}$) to zero while choosing a random element from $\{-1,0,1\}$ for the other elements.

Note that $\mathcal{F}[\cdot]$ in this algorithm denotes the Fourier transform and $\measuredangle z$ denotes the phase of a complex number $z$. 
\begin{algorithm}
  \caption{}
  \label{alg}
    \begin{algorithmic}[1]
    \STATE $Input:\  \ \mathcal{I},\textbf{q}_{1}, \widehat{\textbf{f}}_{1}, (\sigma^{2}_{0},c_1,c_2)>0$.
    \STATE $Output:\ \ \textbf{y}$.
    \FOR {$i=1,...,N$}
     \STATE {\textbf{If} $i\in\mathcal{I}$; $w_{i}\longleftarrow 0$.}
    \STATE {\textbf{If} $i\notin\mathcal{I}$; $w_{i}\longleftarrow \{0,1,-1\} \ w.p.\  \{1/2,1/4,1/4\}$.}    
    \ENDFOR
    \STATE{$\widehat{\textbf{s}}\longleftarrow\textbf{q}_{1}-\widehat{\textbf{f}}_{1}$}
    \STATE{$|\mathcal{F}[\widehat{\textbf{s}}]|\longleftarrow|\mathcal{F}[\widehat{\textbf{s}}]|+a_{1}$; \ $a_{1}\sim\mathcal{N}(0,\sigma_{0}^{2})$.}\\
\STATE{$\measuredangle\mathcal{F}[\widehat{\textbf{s}}]\longleftarrow\measuredangle\mathcal{F}[\widehat{\textbf{s}}]+a_{2}$;\  $a_{2}\sim$ Uniform$[0,\pi/(2Kc_1)]$.}
    \STATE{$\textbf{y}\longleftarrow\mathcal{F}^{-1}[\mathcal{F}[\widehat{\textbf{s}}]]+c_2\textbf{w}$.}
      \end{algorithmic}
\end{algorithm} 

Section \ref{simopf} demonstrates the performance of this attack on a OPF code.

\subsection{Column-Wise Generated Fingerprints }\label{CWC}

Each entry of the fingerprinting matrix $\textbf{F}_{N\times M}$ is chosen independently from $\{0,1\}$ with $p(f_{i,j}=1)=p_{j}, j=1,...,M$, where $p_{j}=\sin^{2}(r_{j})$ is selected by picking uniformly at random the value $r_{j}\in [t,\pi/2-t]$ for some fixed $0<t<\pi/4$, $\sin^{2}(t)=a>0$. Note that in this set up, since the alphabet set is $\{0,1\}$, we have $\mathcal{U}=\{-1,1\}$. 

Suppose a coalition w.l.o.g. $\mathcal{K}=\{1,...,K\}$ performs the proposed attack in section \ref{att} on this structure. Let $\mathcal{I}\subseteq\{1,...,N\}$ be the set of all indices that the colluders could estimate their fingerprints correctly. For this structure, $\mathcal{I}$ is the index set of those rows in $\textbf{A}:=[\textbf{q}_{1}-\textbf{q}_{1},...,\textbf{q}_{1}-\textbf{q}_{K}]$ that have at least one non-zero entry. \\
Next step, the colluders estimate their corresponding $p_{j}$s using empirical estimation over rows with index set $\mathcal{I}$, namely
\begin{equation}
\widehat{p}_{j}:=\dfrac{n_{j}}{|\mathcal{I}|}, \ \ j\in\mathcal{K},
\end{equation} 
where $n_{j}=|\{f_{i,j}=1: i\in\mathcal{I}\}|$. Theorem \ref{ltardos} describes how accurate this estimation will be.
Finally, the attackers use the above empirical distributions to decide about those rows with indices belonging to $\{1,...,N\}\setminus \mathcal{I}$. To do so, they compute 
$$
p_{tot}:=\frac{1}{K}\sum_{j\in \mathcal{K}}\widehat{p}_{j}.
$$
Then for $ i\in\{1,...,N\}\setminus \mathcal{I}$ and $ j\in\mathcal{K}$,
$$
\widehat{f}_{i,j} =
\begin{cases}
1, & \text{if }p_{tot}>0.5+\tau\\
0, & \text{if }p_{tot}<0.5-\tau
\end{cases}
$$
 In case of $0.5-\tau\leq p_{tot}\leq0.5+\tau$, they pick $n_{*}$ indices out of $\{1,...,N\}\setminus \mathcal{I}$ randomly and set them to be one and the rest to be zero, where $n_{*}$ is chosen such that the empirical distributions do not vary, namely
 $$
 \left|\min_{j}\widehat{p}_{j}-\frac{n_{*}+|\mathcal{I}|\min_{j}\widehat{p}_{j}}{N}\right|= \left|\max_{j}\widehat{p}_{j}-\frac{n_{*}+|\mathcal{I}|\max_{j}\widehat{p}_{j}}{N}\right|.
 $$
Equivalently,
$$
n_{*}:=\left\lfloor\frac{\min_{j}\widehat{p}_{j}+\max_{j}\widehat{p}_{j}}{2}(N-|\mathcal{I}|)\right\rfloor.
$$
The final forgery in this attack will be 
$$
\textbf{y}=\textbf{q}_{1}-\widehat{\textbf{f}}_{1}.
$$
\begin{theorem}\label{ltardos}
Consider a probabilistic fingerprinting code $\textbf{F}$ of size $N\times M$ as it is described above and let $\mathcal{K}$ be a coalition of size $K$. 
Then
\begin{itemize}
\item the number of rows that they fail to estimate correctly ($N-|\mathcal{I}|$) is bounded by $N/K$ with probability at least 
$$
1-12K^2\left(\dfrac{3}{8}\right)^K-\dfrac{8K^2}{N2^K}, \ K>6.
$$
\item $\mathbb{E}[||\textbf{f}_{1}-\widehat{\textbf{f}}_{1}||^{2}]\leq N/2^{K-1}$.
\item with probability at least $1-\dfrac{2K}{N^{(2-1/2^{K-2})}}$, we have
$$
\max_{j\in\mathcal{K}}|p_{j}-\widehat{p}_{j}|<\sqrt{\frac{\log N}{N}}.
$$

\end{itemize}
\end{theorem}

\begin{proof}
Proof is in Appendix.
\end{proof}

\section{Attack on random Gaussian structure}\label{contin}
In this section, we generalize the proposed attack for the case that the alphabet set is continuous. 

Suppose the entries of the fingerprinting matrix $\textbf{F}_{N\times M}$ are chosen i.i.d. at random from a normal distribution $\mathcal{N}(0,\sigma^{2})$, such that (\ref{energy}) holds i.e. $\sigma^{2}=1/N$. The main idea here is to convert this fingerprinting structure into a finite alphabet structure by quantizing the support set of the fingerprints and then apply the attack strategy of Section \ref{att}.

Using the Chebyshev's inequality, the attackers can assume that there exists a $\xi>0$ such that the fingerprinting entries, $\{f_{k,m}\}$s are chosen only from a bounded interval $[-\xi,\xi]$ with high probability. This is a good assumption since for any $\epsilon>0$, there exists a $\xi>0$, such that 
\begin{equation}\nonumber
p\left(|f_{i,j}|\geq\xi\right)\leq \epsilon.
\end{equation}
 Furthermore, the attackers discretize the values in this interval into $2w, w\in \mathbb{N}$ discrete values $\{\widehat{\xi}_{i}\}_{i=1}^{2w}$ as follows:  they assign $\widehat{\xi}_{i}=-\xi+\frac{\xi}{2w}(2i-1)$ to all the values that belong to interval $[-\xi+\frac{\xi}{w}(i-1),-\xi+\frac{\xi}{w}i)$ for some $i\in\{1,...,2w\}$ and use  $\{\widehat{\xi}_{i}\}_{i=1}^{2w}$ as their alphabet set. Their next step is to represent the entries of the matrix $\textbf{A}$ introduced in Section \ref{att} (whose $i$th column is given by $\textbf{q}_{1}-\textbf{q}_{i}$) using the new alphabet set i.e., $\{\widehat{\xi}_{i}\}_{i=1}^{2w}$. This can be done using the fact that each entry of the matrix $\textbf{A}$ will fall into $[\frac{\xi}{w}(i-j-1),\frac{\xi}{w}(i-j+1))$, for some $i,j\in\{1,...,2w\}$ with probability at least $(1-\epsilon)^{2}$. The attackers estimate $\textbf{A}=[a_{k,m}]$ by $\widehat{\textbf{A}}_\alpha=[\widehat{a}_{k,m}]$, where 
\begin{equation}\label{quan}
\widehat{a}_{k,m}=
\begin{cases}
-\frac{\xi}{w}(2w-1), & \text{if } a_{k,m}<-2\xi\vspace{1.5mm}\\
\frac{\xi}{w}(i-j), & \text{if } a_{k,m}\in \small{[\frac{\xi}{w}(i-j-1),\frac{\xi}{w}(i-j))} \\
 &\ \ \ \ \ for \ \small{ i<j}\\
 -\frac{\xi}{w}, & \text{if } a_{k,m}\in \small{[-\frac{\xi}{w},-\alpha\frac{\xi}{2w})}\vspace{1.5mm} \\
0, & \text{if }|a_{k,m}|\leq \small{\alpha\frac{\xi}{2w}}\vspace{1.6mm}\\
 \frac{\xi}{w}, & \text{if } a_{k,m}\in \small{(\alpha\frac{\xi}{2w},\frac{\xi}{w}]}\vspace{1.5mm} \\
\frac{\xi}{w}(i-j), & \text{if } a_{k,m}\in \small{(\frac{\xi}{w}(i-j),\frac{\xi}{w}(i-j+1)]}\\
 &\ \ \ \ \ for \ \small{i>j}\\
\frac{\xi}{w}(2w-1), & \text{if } a_{k,m}>2\xi
\end{cases}
\end{equation}
for some $i,j\in\{1,...,2w\}$ and $0<\alpha\leq2$. Next we provide an example for the case when $\xi=2$ and $w=2$ to motivate our quantization choices.

\begin{example}
\textit{Let $\xi=2$ and $w=2$. The attacker's alphabet will be $\{-3/2,-1/2,1/2,3/2\}$. If $a_{k,m}:=f_{k,1}-f_{k,m}\in [-4,-3)$, then with probability at least $(1-\epsilon)^{2}$, we have $-2<f_{k,1}<-1$ and $1<f_{k,m}<2$. In other word, $\hat{f}_{k,1}=-3/2$ and $\hat{f}_{k,m}=1/2$, equivalently, $\widehat{a}_{k,m}=-3/2-1/2=-3$. Figure \ref{quantize} illustrates this quantization.\\
Now, suppose $a_{k,m}$ belongs to $(-3,-2)$. In this case, there are three possible events: $E_{1}:=\{-2<f_{k,1}<-1, 1<f_{k,m}<2\}$, $E_{2}:=\{-2<f_{k,1}<-1, 0<f_{k,m}<1\}$, or $E_{3}:=\{-1<f_{k,1}<0, 1<f_{k,m}<2\}$. If $E_{1}$ is the true event, then $\widehat{a}_{k,m}=-3/2-1/2=-3$. However, if either $E_{2}$ or $E_{3}$ is the true event, then $\widehat{a}_{k,m}=-1/2-3/2=-2$. Since $p(E_{2}\cup E_{3})>p(E_{1})$, we encode $\widehat{a}_{k,m}=-2$.
The quantization rule given by (\ref{quan}) is obtained similarly, except for the case that $a_{k,m}\in [-1,1)$. In this case, we consider a small gap interval controlled by $\alpha$. When $\alpha=0$, $\widehat{a}_{k,m}=0$ is neglected and when $\alpha=2$, the whole interval $[-1,1]$ is quantized into $0$.
}
\begin{figure}
  \centering
    \includegraphics[width=0.4\textwidth]{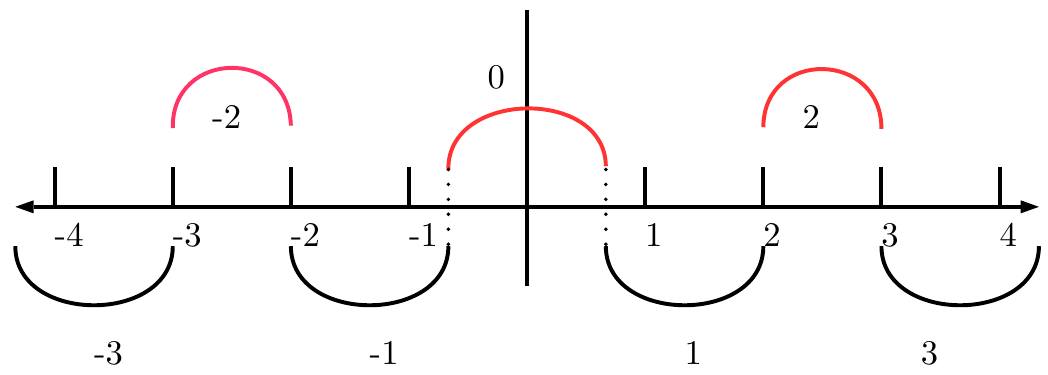}
\caption{Quantization rule for the entries of the matrix $\textbf{A}$, when $\xi=2$, $w=2$, and $\alpha=1.4$  .}\label{quantize}
\end{figure}
\end{example}

Without loss of generality, suppose the first $K$ users are involved in the attack and $\widehat{\textbf{f}}_{i}$ is the attacker's estimate of $\textbf{f}_{i}$ for $1\leq i\leq K$. Then their final forgery is
\begin{equation}\label{forgery}
\widehat{\textbf{y}}=\frac{1}{K}\sum_{i=1}^{K}(\textbf{q}_{i}-\widehat{\textbf{f}}_{i})+\textbf{w},
\end{equation} 
where $\textbf{w}$ is a noise vector with mean zero and variance $\sigma_{0}^{2}$ introduced by the attackers in order to hide their tracks. 
Note that one can easily find an upper bound on the expected estimation error of this attack as follows:
\begin{eqnarray}\nonumber
&\mathbb{E}\left[||\textbf{f}_{1}-\widehat{\textbf{f}}_{1}||^{2}\right]^{1/2}\leq\sqrt{N\max_{j} \mathbb{E}[|f_{1,j}-\widehat{f}_{1,j}|^{2}]}\leq\\ \nonumber
 &\sqrt{N}\left(\int_{0}^{\infty}\dfrac{\left(x+(2w-1)\xi/2w\right)^{2}}{\sqrt{2\pi\sigma^{2}}} e^{-\frac{x^{2}}{2\sigma^{2}}}dx\right)^{1/2}\leq\\ \nonumber
 &\sqrt{\frac{N}{2}}\left(\sigma+\frac{(2w-1)\xi}{2w}\right).
\end{eqnarray}
Since $\sigma=1/\sqrt{N}$, by letting $\xi=O(1/\sqrt{N})$, the above error is bounded as $N$ grows. Later in Section \ref{simulation}, we study the performance of this attack when $w=2$ and compare its performance with the uniform linear averaging attack which was proved to outperform order-statistic collusion attacks on random Gaussian fingerprints when the colluders are subject to a mean-squared distortion constraint \cite{nn3}.
The simulation result demonstrates the excellent performance of the proposed algorithm even in continuous case.

\section{Simulation Results}\label{simulation}
Herein, we simulate the performance of our attack on both random and deterministic structures.

\subsection{Random Ternary Fingerprints}
We attacked two RTF systems with probability vectors $\textbf{p}_{1}=(2/3,1/6)$ and $\textbf{p}_{2}=(1/3,1/3)$ each for three different dimensions $(N,M)\in\{(729,2016),(2187,8128),(8128,16384)\}$. Figure \ref{simu2} depicts the probability of failing to estimate the signal  $\textbf{s}$ correctly as a function of the number of colluders for the random ternary structures. The failure event is defined as estimating at least $1\%$ of the host signal incorrectly. As Figures \ref{simu2} illustrates that the number of colluders necessary to  estimate the host signal in RTF system is much less than what Lemma \ref{lemma22} suggests. This confirms the effectiveness of our proposed attack in practice beyond the theoretical limit.

\begin{figure}
\hspace{-3mm}
\includegraphics[width=95mm, height=48mm]{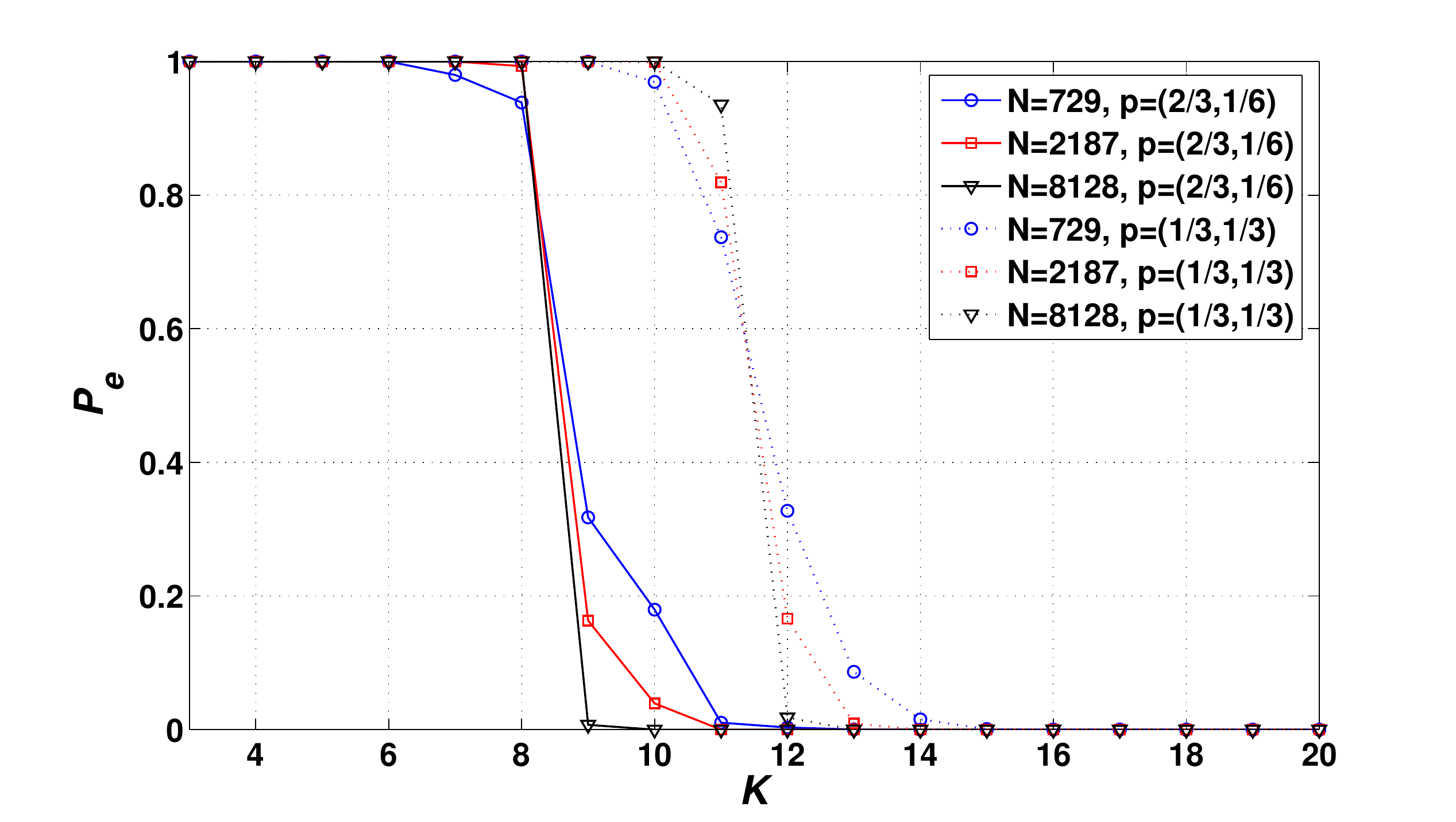}
\caption{Probability of estimating at least $1\%$ of the host signal incorrectly as a function of the collusion size $K$ for RTF structure.}\label{simu2}
\end{figure}

\subsection{ETF Fingerprints}
We attacked two ETF systems with matrices  $\textbf{F}_{1}=\textbf{F}_{ETF}(2,7,91,16)$ and $\textbf{F}_{2}=\textbf{F}_{ETF}(2,2,2^7,128)$, respectively.  Figure \ref{simu3} depicts the probability of failure event for these two ETF structures, when the failure event is defined as at least estimating one coordinate of the host signal $\textbf{s}$ incorrectly. Again the number of users that succeed at estimating the host signal is much less than what the bounds in Theorem \ref{attacktheorem} suggests.
\begin{figure}
\hspace{-4mm}
\includegraphics[width=95mm, height=48mm]{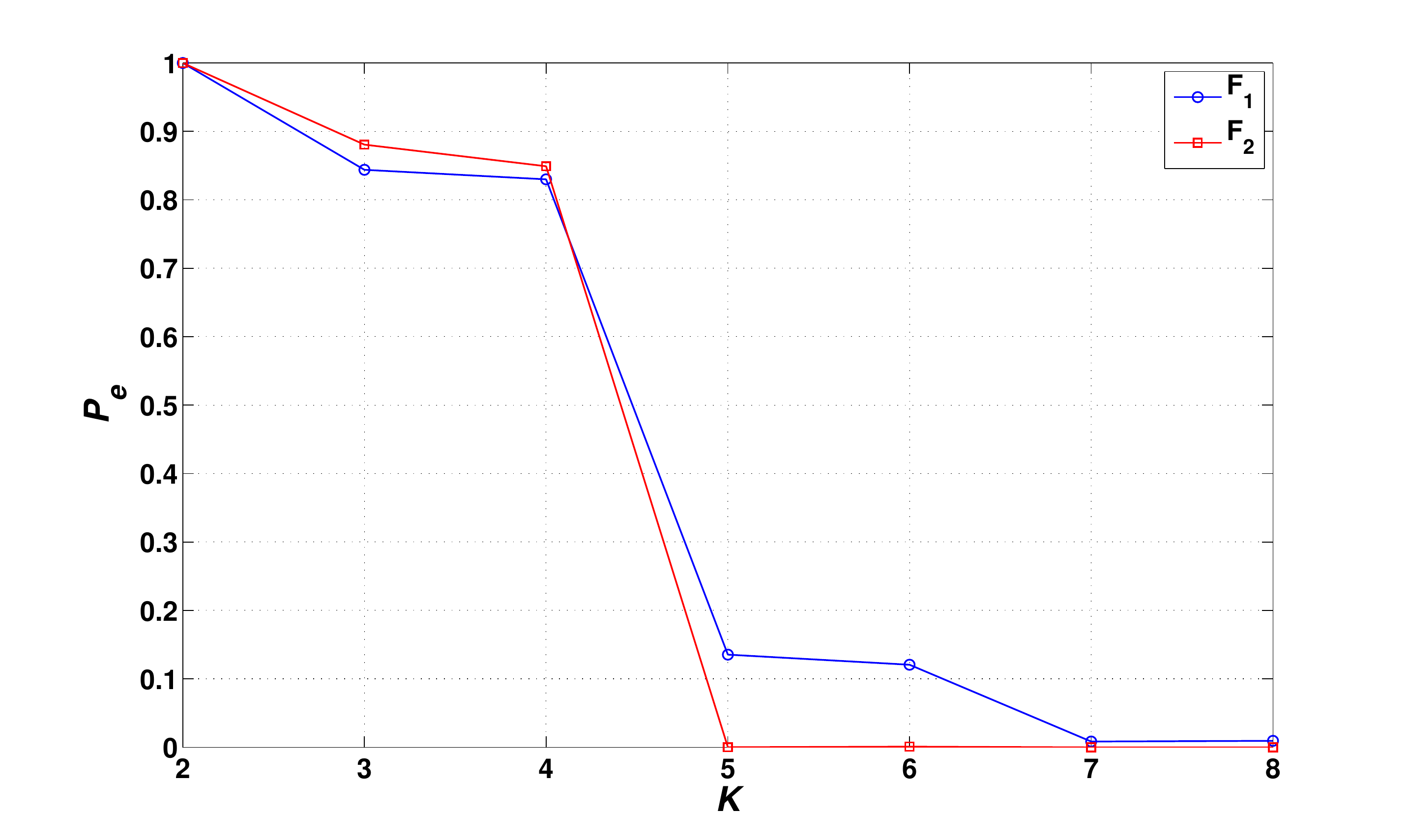}
\caption{Probability of $||\textbf{s}-\hat{\textbf{s}}||>0$ as a function of collusion size $K$ for ETF structure.}\label{simu3}
\end{figure}

\subsection{Gaussian Fingerprints}
We also attacked two random Gaussian fingerprinting systems with dimensions $(N,M)\in\{(100,250), (1000,2070)\}$. The attackers parameters were set at $(\xi,w,\alpha)=(0.12,2,1)$ and the forgery was done according to (\ref{forgery}).   Figure \ref{er1} depicts $||\textbf{f}_{i_{1}}-\widehat{\textbf{f}}_{i_{1}}||$, the estimation error of  the attackers $i$,  as a function of the collusion size $K$. 

We compared the performance of our attack against the uniform linear averaging attack, the optimal attack in the class of order-statistic collusion attacks on random Gaussian fingerprints when the colluders are subject to a mean-squared distortion constraint \cite{nn3}. The colluders in that attack create their forgery by uniformly averaging their copies
and adding an i.i.d. Gaussian noise. We evaluated performance of both attacks for a focused fingerprint detector, which performs a binary hypothesis test for each user to decide whether that particular user is guilty \cite{ETF}.

Figures \ref{comparing15} and \ref{comparing2} show the probability of detecting at least one colluder 
as a function of the number of colluders $K$ using the focused detector.

\begin{figure}[h]
\hspace{-2mm}
\includegraphics[width=95mm, height=48mm]{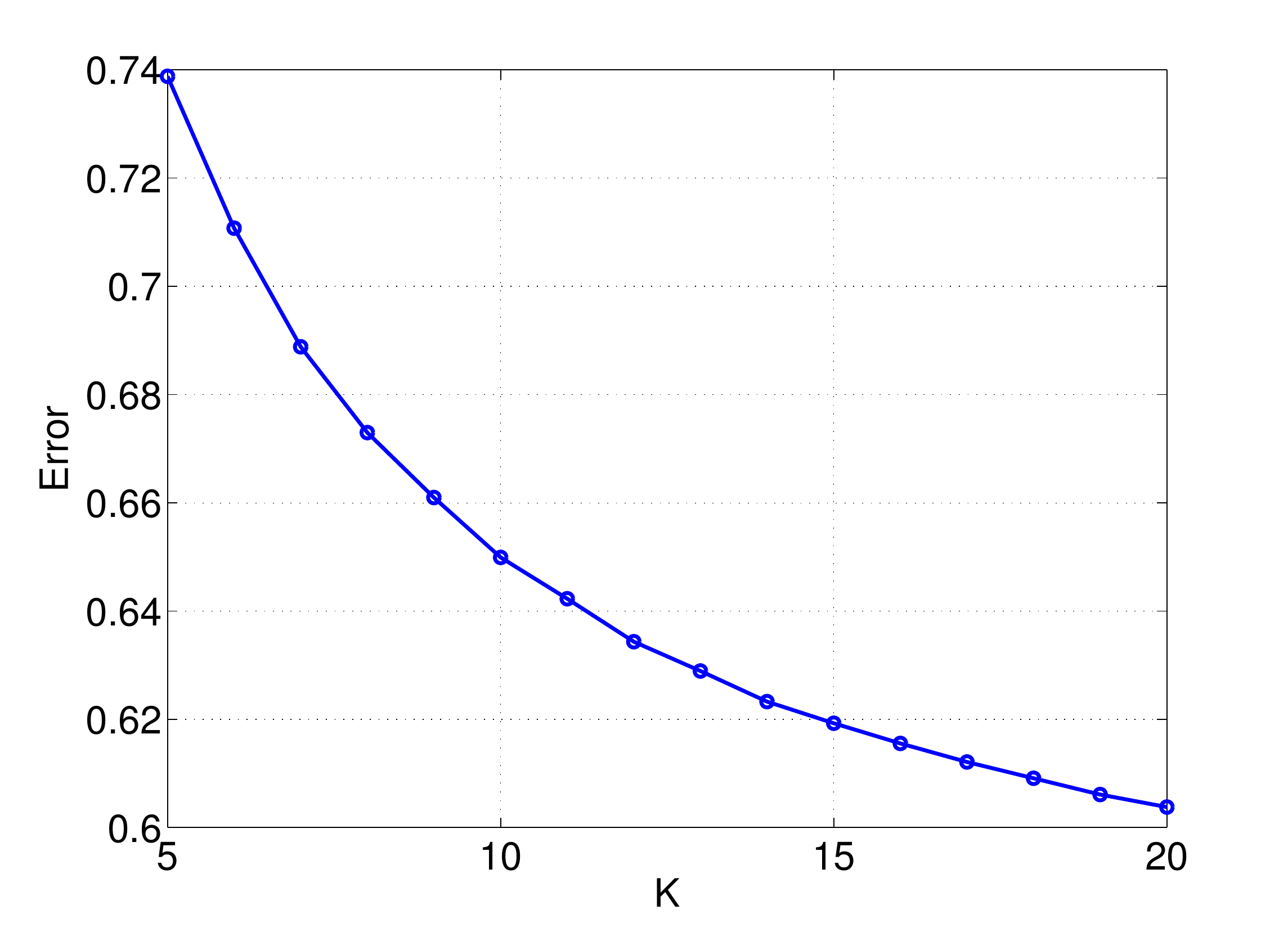}
\caption{Estimation error  $||\textbf{f}_{i_{1}}-\widehat{\textbf{f}}_{i_{1}}||$ for random Gaussian fingerprints ($(N,M)=(1000,2070)$ ).}\label{er1}
\end{figure}

\begin{figure}[h]
\hspace{-2mm}
\includegraphics[width=95mm, height=48mm]{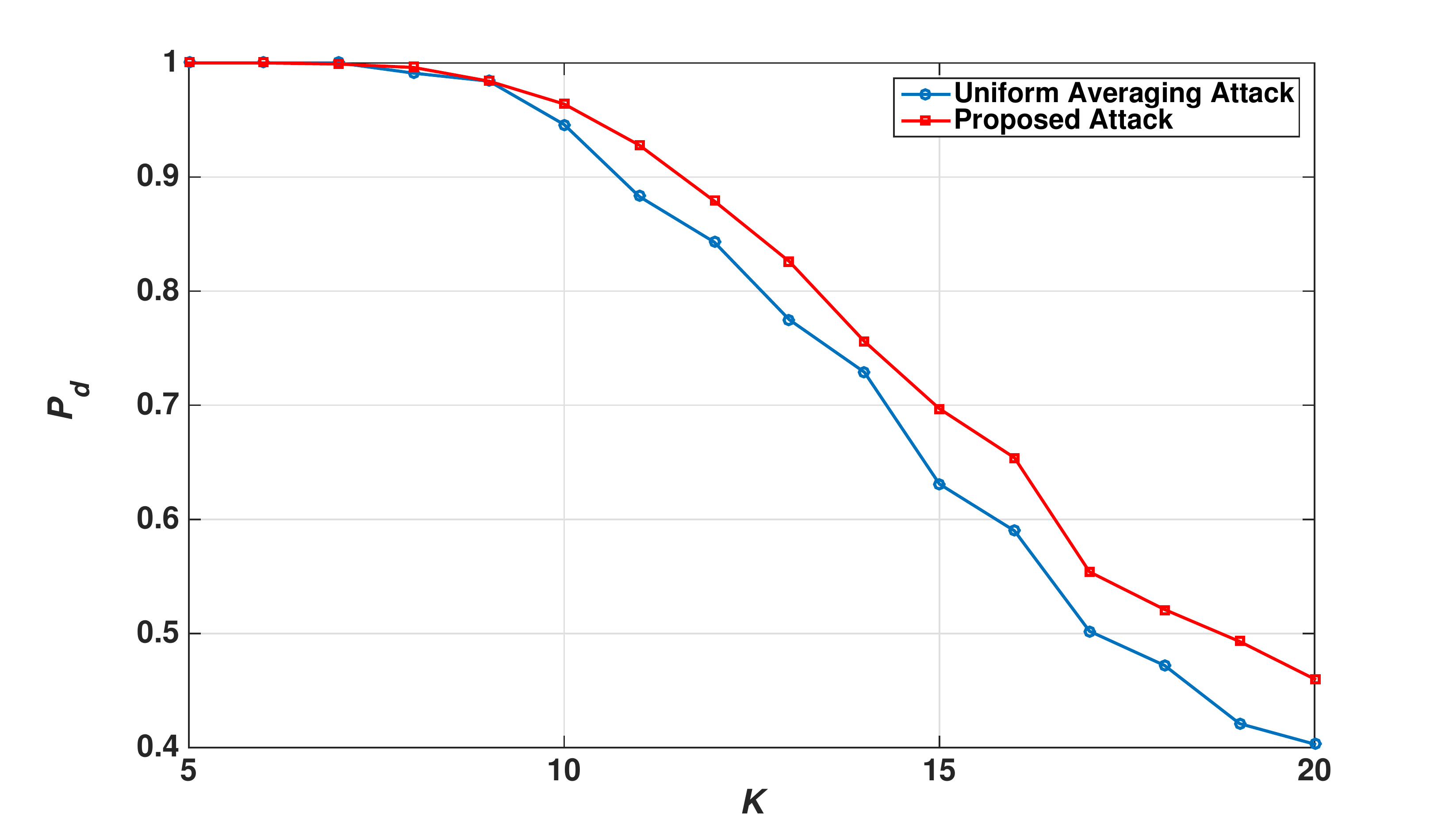}
\caption{A comparison of probability of detecting at least one colluder for our proposed attack vs. uniform linear averaging attack as a function of number of colluders for random Gaussian fingerprints ($(N,M)=(1000,2070)$). }\label{comparing15}
\end{figure}

\begin{figure}[h]
\hspace{-2mm}
\includegraphics[width=95mm, height=48mm]{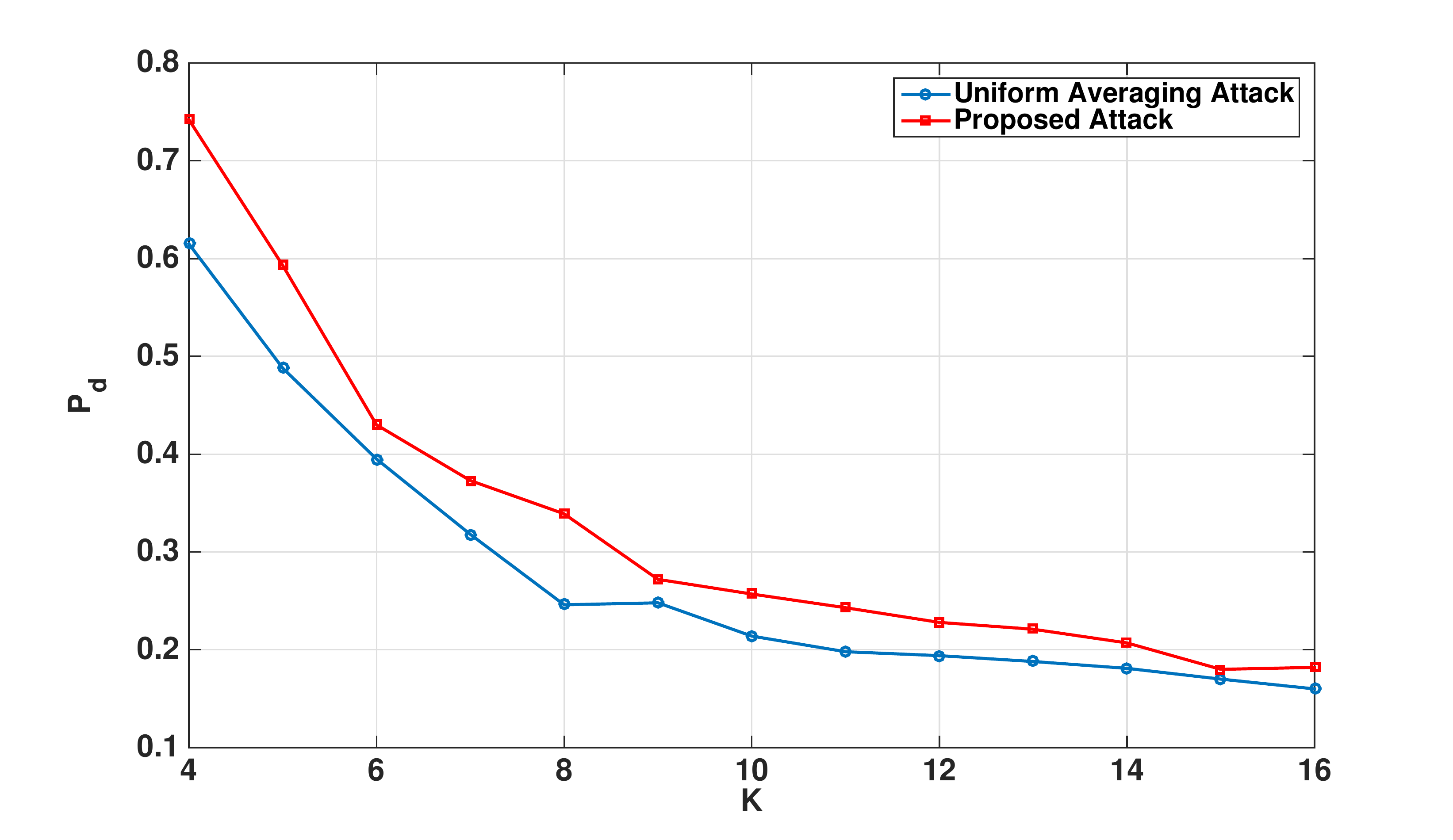}
\caption{A comparison of probability of detecting at least one colluder for our proposed attack vs. uniform linear averaging attack as a function of number of colluders for random Gaussian fingerprints ($(N,M)=(100,250)$).}\label{comparing2}
\end{figure}

\subsection{Tardos-like (Column-Wise Random) fingerprints}
The performance of the proposed attack against the fingerprinting code proposed in Section \ref{CWC} is studied in this section. We generated a fingerprinting code with parameters $(N,M,t')=(5600,2100,\pi/1000)$. $K\in\{4,...,11\}$ number of colluders were selected randomly $500$ times to arrange the attack against this code. 

$K\in \{4,...,11\}$ number of colluders were selected randomly 500 times to carryout the attack. Figure \ref{cwc1} demonstrates the worst estimation error rate incurred by an attacker, i.e., $\max_{j\in\mathcal{K}}||\textbf{f}_{j}-\widehat{\textbf{f}}_{j}||^{2}/N$ along with the theoretical bound introduced in Theorem \ref{ltardos} as a function of number of colluders. The plot suggests that in practice the performance of our attack is considerably better than predicted by the theoretical  bound. 

Additionally, we evaluated the performance of our attack against the focused detector which aims to identify at least one of the colluders. We selected $K
$ attackers and performed the attack 500 times at random. Figure \ref{cwc3} illustrates the probability of catching at least one colluder as a function of number of colluders. In the figure, probability of catching $P_c$ is calculated as
\begin{equation}\label{pd}
P_c:=\frac{\# \text{trials that at least one colluders is detected}}{\# \text{total trials}}.
\end{equation}

\begin{figure}[h]
\hspace{-2mm}
\includegraphics[width=95mm, height=48mm]{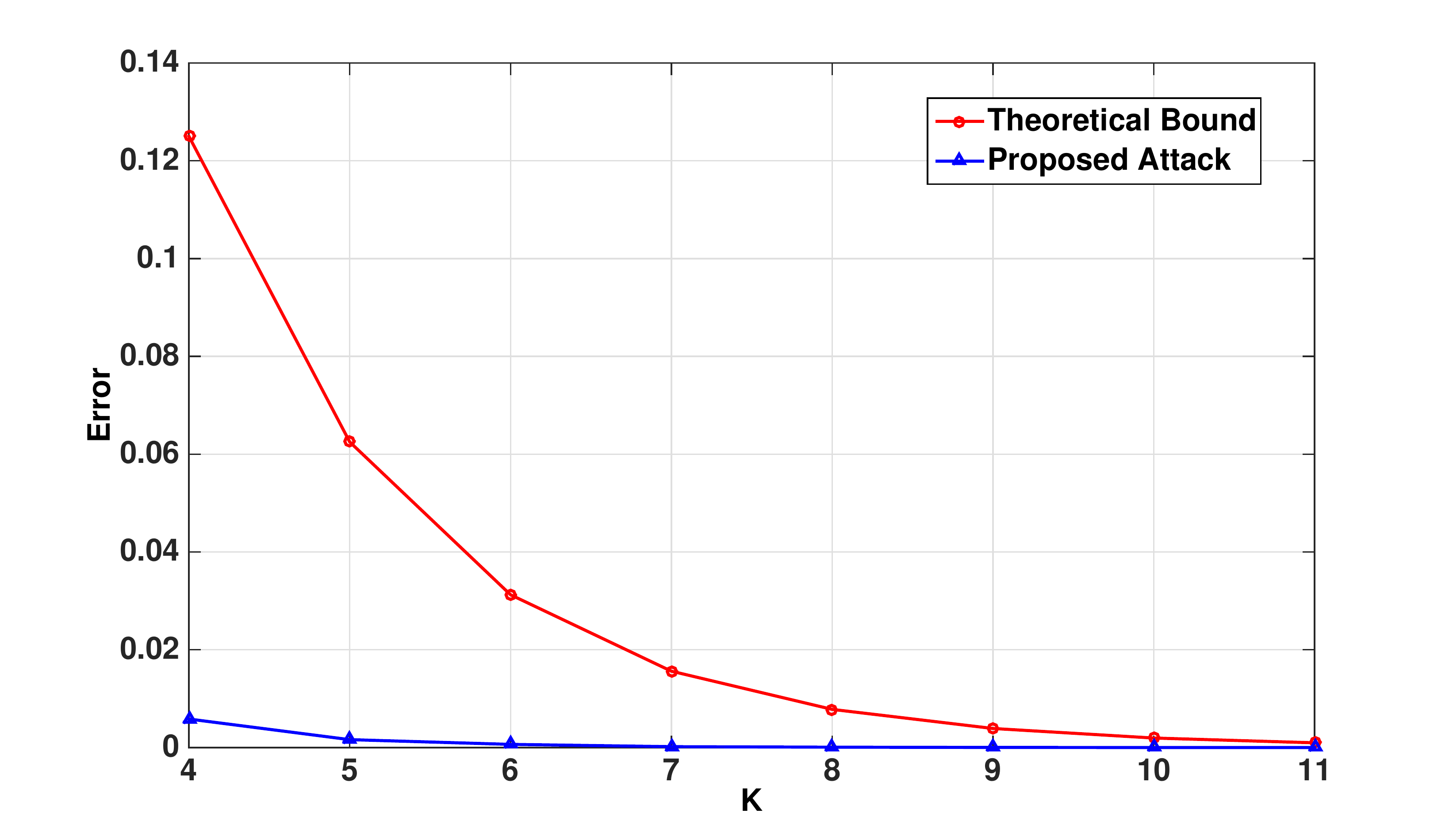}
\caption{Worst estimation error rate of the fingerprints and its corresponding theoretical bound for the CWC fingerprinting structure. }\label{cwc1}
\end{figure}

\begin{figure}[h]
\hspace{-2mm}
\includegraphics[width=95mm, height=48mm]{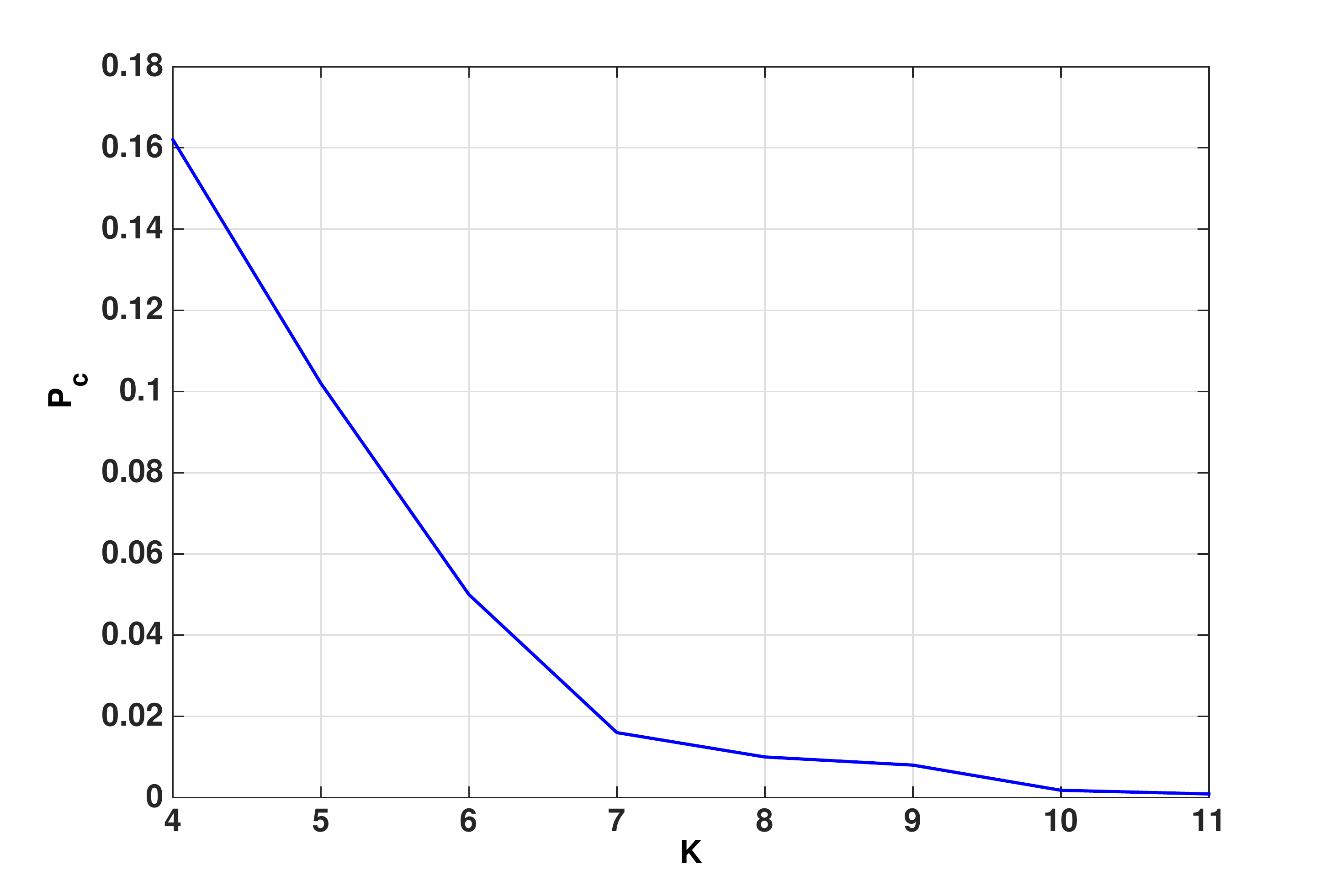}
\caption{A plot of the probability of catching at least one colluder $P_{c}$ as a function of the number of colluders for the CWC structure.}\label{cwc3}
\end{figure}

\subsection{Row-Wise Generated Fingerprints (Tardos Code)}\label{simopf}

We simulated our proposed attack against the optimal probabilistic fingerprints (OPF), known as Tardos codes, introduced in Section \ref{tard}. We considered a length $N=7500$ OPF code designed with parameters $(k,\epsilon)=(5,0.1)$ for $M=1500$ number of users. Our first measure of success of the attack is the size of the set of fingerprint indices the colluders detected correctly, $|\mathcal{I}|$. Figure \ref{correction} demonstrates the rate $|\mathcal{I}|/N$ as a function of the
coalition size for $K=\{2,3,4,5\}$.

\begin{figure}[h]
\hspace{-5mm}
\includegraphics[width=95mm, height=48mm]{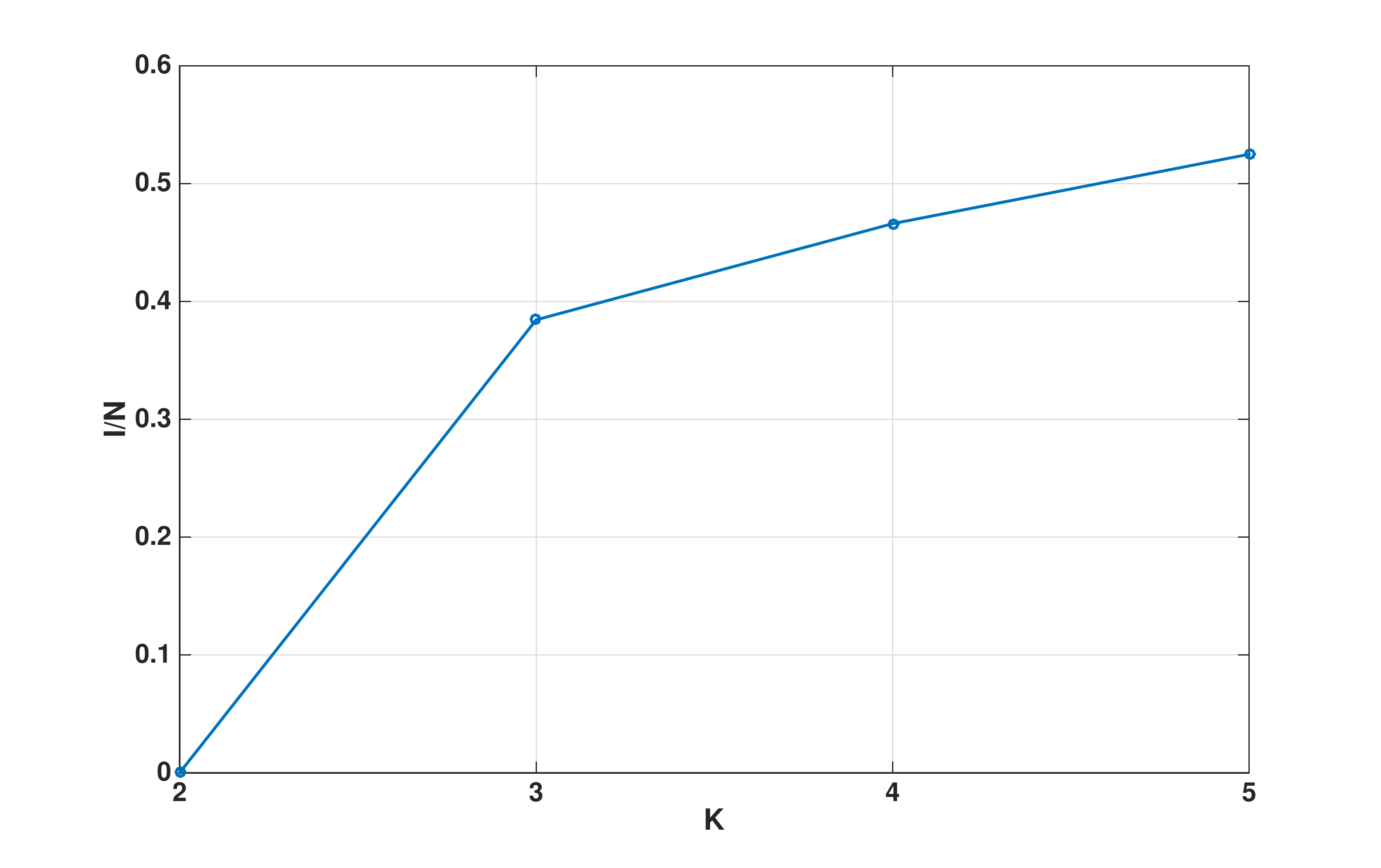}
\caption{A plot of $|\mathcal{I}|/N$ as a function of the number of colluders for the Tardos code with parameter $(k,\epsilon)=(5,0.1)$.}\label{correction}
\end{figure}

Next, we selected $K=5$ users randomly and attacked the fingerprinting scheme using  the forgery described by Algorithm \ref{alg} for $3000$ trials. We compared the performance of our attack against the uniform averaging, minority voting, and majority voting attacks \cite{minority, minority2}. 
In minority voting attack, the attackers output a symbol that occurs less often than all
the others in their copies. Similarly, in majority voting, the attackers output the most frequent symbol \cite{gaussian_aprox}. All attacks were evaluated against the optimal accusation detector of Tardos \cite{tardos}, which was described in Section \ref{tard}.
We computed the false positive (FP) probability, i.e.,  the probability of accusing at least one innocent user, as a function of the ratio of the fingerprint power to the noise power introduced by the colluders (Ratio) for all 3 attacks. More precisely, 
$$
\text{Fingerprint to Noise Power Ratio (FNPR)}:=20\log_{10}\frac{||\textbf{f}_1||}{||\textbf{y}-\textbf{s}||},
$$
where $\textbf{y}$ is the forgery. As Figure \ref{PF} illustrates, the proposed attack outperforms the others since for the same number of colluders and the same FNPR, it results in a higher false positive. 

\begin{figure}
\hspace{-5mm}
\includegraphics[width=95mm, height=50mm]{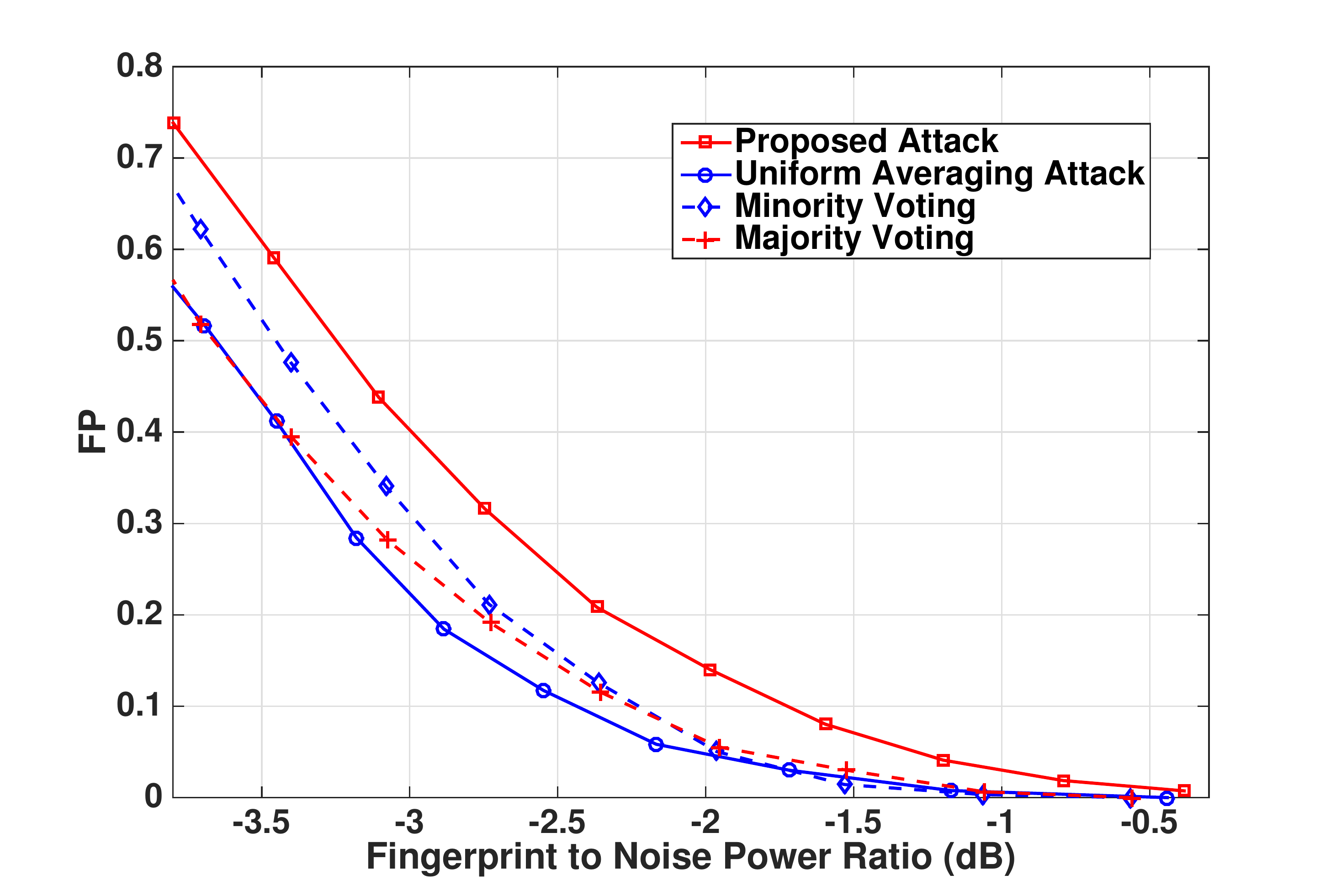}
\caption{Comparison of the fales positive (FP) probabilities of different attacks against the Tardos code with parameter $(k,\epsilon)=(5,0.1)$.}\label{PF}
\end{figure}

\section{Conclusion}\label{con}
 We develop a new attack on fingerprinting system with finite alphabet fingerprints. We derive theoretical bounds on the number of colluders that can defeat the fingerprinting system. Through simulation, we show  that a much smaller number of attackers can indeed defeat the system in practice. Our analysis suggests that the designer of a random fingerprinting system with parameters ($\Xi,\textbf{p}_{\Xi}$) should select the fingerprints from an alphabet $\Xi$ whose associated uniquely decodable set, $\mathcal{U}$, has a small cardinality. Moreover, the fingerprinter should pick a source distribution $\textbf{p}_{\Xi}$ that maximizes the error probability (\ref{Error}) subject to the energy constraint (\ref{energy}).

\section*{Acknowledgment}

\section{Appendix}\label{appen}
\textit{Proof of Theorem \ref{attacktheorem}}\\
Assume $K$ colluders $\{i_{1},...,i_{K}\}$ take part in the forgery. Let $E_{i}$ to be the event that the attackers estimate $s_{i}$ incorrectly, i.e., $E_{i}:=\{\hat{s}_{i}\neq s_{i}\}$. Since the attackers know that the fingerprint matrix is sparse, when they cannot estimate $f_{i,i_{1}}$\footnote{Note that for the attackers, estimating at least one of the fingerprints is equivalent to estimating all of their fingerprints.} uniquely  (i.e., $\mathcal{A}_{i}\cap\mathcal{U}=\emptyset$), 
the performance of the proposed attack is lower bounded by the following attack. They estimate $f_{i,i_{1}}$ with an element of $\mathcal{B}\left((f_{i,i_{1}}-f_{i,i_{1}},...,f_{i,i_{1}}-f_{i,i_{K}})\right)$ with the maximum number of zeros. The error event $E_{i}$ only occurs when vector $(f_{i,i_{1}},...,f_{i,i_{K}})$ both does not contain the pair $(-1,+1)$ and the number of its non-zero elements exceeds the number of its zeros. Since each row of $\textbf{F}_{ETF}(r,h,n,m_{0})$ has exactly $hm_{0}$ non-zero elements, and each row of the corresponding Hadamard matrix has equal number of $+1$ and $-1$,  each row of $\textbf{F}_{ETF}(r,h,n,m_{0})$ has exactly $hm_{0}/2$, $-1$s. Therefore, we can upper bound the error event as
\begin{equation}\label{pei}
p(E_{i})\leq2\sum_{i\leq \lfloor K/2\rfloor} {M-hm_{0} \choose i}{hm_{0}/2 \choose K-i}/{M \choose K}.
\end{equation}
Next we introduce an upper bound on the above error probability, first we show that for $K\geq2$ and $n>8h$,
\begin{eqnarray}\label{appe1}
\sum_{i\leq \lfloor K/2\rfloor} {M-hm_{0} \choose i}{hm_{0}/2 \choose K-i}\\ \nonumber
\leq 2{M-hm_{0} \choose \lfloor K/2\rfloor}{hm_{0}/2 \choose \lfloor (K+1)/2\rfloor}.
\end{eqnarray}
The proof is by induction on $K$. The case $K=2$ is trivial since $M>8hm_{0}$. Suppose that the inequality holds for $K<k$. In order to prove the inequality for $K=k$ we consider two cases: $k$ is an even ($k=2q$) or an odd number ($k=2q+1$). Here we only prove the even case since the proof for the odd case is similar. If $k=2q$, we need to show
\begin{eqnarray}\nonumber
\sum_{i\leq q-1} {M-hm_{0} \choose i}{hm_{0}/2 \choose 2q-i}\leq{M-hm_{0} \choose q}{hm_{0}/2 \choose q}.
\end{eqnarray}
This is true because
\begin{eqnarray}\nonumber
&\sum_{i\leq q-1} {M-hm_{0} \choose i}{hm_{0}/2 \choose 2q-i}  \\ \nonumber
&=\sum_{i\leq q-1} {M-hm_{0} \choose i}{hm_{0}/2 \choose 2q-2-i}\dfrac{\small{hm_{0}/2-2q+i+1}}{2q-i-1}\\ \nonumber
&\times\dfrac{hm_{0}/2-2q+i+2}{2q-i}  \\ \nonumber
\end{eqnarray}
\begin{eqnarray}\nonumber
&\leq\sum_{i\leq q-1} {M-hm_{0} \choose i}{hm_{0}/2 \choose 2q-2-i}\dfrac{\small{hm_{0}/2-q}}{q}\\ \nonumber
&\times\dfrac{hm_{0}/2-q+1}{q+1}  \\ \nonumber
&\leq2 {M-hm_{0} \choose q-1}{hm_{0}/2 \choose q-1}\dfrac{(hm_{0}/2-q)(hm_{0}/2-q+1)}{(q+1)q}  \\ \nonumber
&=2 {M-hm_{0} \choose q}{hm_{0}/2 \choose q}\dfrac{(hm_{0}/2-q)q}{(q+1)(M-hm_{0}-q+1)}  \\ \nonumber
&\leq {M-hm_{0} \choose q}{hm_{0}/2 \choose q}
\end{eqnarray}
The first inequality is due to the fact that $g(x):=\dfrac{a-b+x}{b-x+1}$ is an increasing function of $x$ if $a+1>0$. The second inequality is because of the induction hypothesis, and the last inequality is valid since $q/(q+1)\leq1$ and because of the assumption that $M\geq 8hm_{0}$. Substituting (\ref{appe1}) into (\ref{pei}) implies
\begin{eqnarray}\nonumber
&p(E_{i})\leq 4\dfrac{{M-hm_{0} \choose \lfloor K/2\rfloor}{hm_{0}/2 \choose \lfloor (K+1)/2\rfloor}}{{M \choose K}}\\ \nonumber
&=4{K \choose \lfloor (K+1)/2\rfloor}\prod_{j=1}^{\lfloor (K+1)/2\rfloor}\left(1-\dfrac{M-\lfloor K/2\rfloor-hm_{0}/2}{M-K+j}\right)\\ \nonumber
&\times\prod_{j=1}^{\lfloor K/2\rfloor}\left(1-\dfrac{hm_{0}}{M-\lfloor K/2\rfloor+j}\right)\\ \nonumber
&\leq 4\left(\frac{Ke}{\lfloor K/2\rfloor}\right)^{\lfloor K/2\rfloor}\left(\frac{hm_{0}/2}{M-\lfloor K/2\rfloor}\right)^{\lfloor (K+1)/2\rfloor}\left(\frac{M-hm_{0}}{M}\right)^{\lfloor K/2\rfloor}.
\end{eqnarray}
The last inequality is because ${n \choose n-k}={n \choose k}\leq (ne/k)^{k}$, where $e$ is the Euler's constant and the fact that $1-a/(b+x)$ is an increasing function of $x$ when $a>0, b>0$. By considering $K$ to be even or odd separately, and the assumption that $M>8hm_{0}$ one can show that the last equation in the above inequality can be bounded as
\begin{eqnarray}\nonumber
4\left(\frac{Ke}{\lfloor K/2\rfloor}\right)^{\lfloor K/2\rfloor}\left(\frac{hm_{0}/2}{M-\lfloor K/2\rfloor}\right)^{\lfloor (K+1)/2\rfloor}\left(\frac{M-hm_{0}}{M}\right)^{\lfloor K/2\rfloor}\\ \nonumber
\leq4\left(\frac{ehm_{0}(M-hm_{0})}{M(M-K/2)} \right)^{K/2}.
\end{eqnarray}
Recall that $M=nm_{0}$. The result follows from applying the union bound argument.
\\

\textit{Proof of Lemma \ref{errorpro}}\\
It is clear that this structure satisfies the zero mean condition. Since the elements of fingerprints are chosen uniformly, all the elements in $\mathcal{B}(\tilde{\textbf{b}})$ have the same likelihood    for every $\textbf{b}\in\Xi^{K}$. Hence, the error probability (\ref{Error}) can be written as
$$
p(E)=\frac{1}{(2w+1)^{K}}\sum_{\textbf{b}\in\Xi^{K}}(1-\frac{1}{|\mathcal{B}(\tilde{\textbf{b}})|}).
$$
It can be proven that due to the symmetric structure of $\textbf{F}$, $|\mathcal{B}(\tilde{\textbf{b}})|=t$ when the elements of $\textbf{b}$ have maximum distance $(2w+1-t)/z$, i.e., $\max_{i,j}|b_{i}-b_{j}|=(2w+1-t)/z$. Moreover, $|\mathcal{B}(\tilde{\textbf{b}})|\leq2w+1$. Using these facts we obtain
\begin{eqnarray}\nonumber
\sum_{\textbf{b}\in\Xi^{K}}1-\frac{1}{|\mathcal{B}(\tilde{\textbf{b}})|}=\sum_{t=1}^{2w+1}(1-\frac{1}{t})|\{\textbf{b}:|\mathcal{B}(\tilde{\textbf{b}})|=t\}|\\ \nonumber
=2w+\sum_{t=2}^{2w}\bigg((t-1)\bigg((2w+2-t)^{K}-2(2w+1-t)^{K}\\ \nonumber
+(2w-t)^{K}\bigg)\bigg)=(2w)^K.
\end{eqnarray}
This implies
$
p(E)=\left(1-\frac{1}{2w+1}\right)^K
$ and by applying the union bound, the result is immediate.
\\

\textit{Proof of Lemma \ref{lemma22}}\\
In this scenario, it is clear that the performance of the proposed attack is lower bounded by an attack in which the attackers choose an element in  $\mathcal{B}\left((f_{i,1}-f_{i,1},...,f_{i,1}-f_{i,K})\right)$ with maximum number of zeros whenever $\mathcal{A}_{i}\cap\mathcal{U}=\emptyset$. To analyze this attack, we consider two cases: $0<p\leq 1/4$ and $1/4<p<1/2$. If $0<p\leq 1/4$, it can be shown that the attackers might fail to learn $s_{i}$ only when both $\mathcal{A}_{i}\cap\mathcal{U}=\emptyset$ and the number of non-zero elements exceeds the number of zeros. Let $p(E)$ to be the probability of such event, then
\begin{equation}\label{jhe}
p(E)\leq2\sum_{j\geq\lfloor (K+1)/2\rfloor}{K \choose j}p^{j}(1-2p)^{K-j}.
\end{equation}
From application of binomial expansion,
\begin{eqnarray}\nonumber
&(3/2)^{K}=\sum_{j=0}^{\lfloor (K-1)/2\rfloor}{K \choose j}\dfrac{1}{2^{j}}\\ \nonumber
&+\sum_{j=0}^{\lfloor K/2\rfloor}{K \choose j+\lfloor (K+1)/2\rfloor}\dfrac{1}{2^{j+\lfloor (K+1)/2\rfloor}},
\end{eqnarray}
we obtain
\begin{eqnarray}\label{namosavi}
\sum_{j=0}^{\lfloor K/2\rfloor}{K \choose j+\lfloor (K+1)/2\rfloor}\dfrac{1}{2^{j}}\leq 2^{\lfloor (K+1)/2\rfloor}(3/2)^{K}.
\end{eqnarray}
Rewrite the right hand side of (\ref{jhe}) as follows
\begin{eqnarray}\nonumber
&\small{2p^{\lfloor(K+1)/2\rfloor}(1-2p)^{\lfloor K/2\rfloor}\sum_{j=0}^{\lfloor K/2\rfloor}{K \choose j+\lfloor(K+1)/2\rfloor}\left(\dfrac{p}{1-2p}\right)^{j}}\\ \nonumber
&\overset{(a)}{\leq}\small{2p^{\lfloor(K+1)/2\rfloor}(1-2p)^{\lfloor K/2\rfloor}\sum_{j=0}^{\lfloor K/2\rfloor}{K \choose j+\lfloor(K+1)/2\rfloor}\dfrac{1}{2^{j}}}\\ \nonumber
&\overset{(b)}{\leq} \small{2(2p)^{\lfloor(K+1)/2\rfloor}(1-2p)^{\lfloor K/2\rfloor}(3/2)^{K}}\leq2\left(\dfrac{3}{4}\right)^{K}.
\end{eqnarray}

Inequality $(a)$ follows because $p\leq1/4$, thus $2p\leq1-2p$. Inequality $(b)$ follows from (\ref{namosavi}) and the last inequality is due to the fact that $\max\{x^{q+1}(1-x)^{q},x^{q}(1-x)^{q}\}\leq \dfrac{1}{4^{q}}$, when $0\leq x\leq 1-x,\ q\in\mathbb{N}$.

If $1/4<p<1/2$, we assume the attackers make an error any time $\mathcal{A}_{i}\cap\mathcal{U}=\emptyset$ or equivalently when $(f_{i,1},...,f_{i,K})$ does not include the pair $(-1,1)$. This happens with probability less than $2(1-p)^{K}-(1-2p)^{K}$ which is also less than $2\left(3/4\right)^{K}$ for $1/4<p<1/2$. Hence, $p(E)\leq2\left(3/4\right)^{K}$. The result is immediate from the union bound.
\\


\textit{Proof of Theorem \ref{tardos2}}\\
Let $S=\sum_{i=1}^{N}\mathbb{I}_{\{\textbf{A}_{i}=[0,...,0]\}}$. Hence, $S$ is the number of rows that colluders can not estimate their corresponding fingerprints correctly, and let $p(e_{i})$ to be the probability of event $\{\mathbb{I}_{\{\textbf{A}_{i}=[0,...,0]\}}=1\}$. We obtain
\begin{eqnarray}\label{exp}
&\mathbb{E}[S]=\mathbb{E}[\mathbb{E}[S|e_{i}, 1\leq i\leq N]]\\ \nonumber
&=\sum_{i=1}^{N}\mathbb{E}[p(e_{i})]=\sum_{i=1}^{N}\mathbb{E}[\varrho^{K}_{i}+(1-\varrho_{i})^{K}]\\ \nonumber
&=N\left(\int_{t}^{\pi/2-t}\sin^{2K}(x)+\cos^{2K}(x)dx\right)/(\pi/2-2t).
\end{eqnarray}
Using the fact that
\begin{eqnarray} \nonumber
\int\sin^{n}(x)dx=-\frac{\sin^{n-1}(x)\cos(x)}{n}+\frac{n-1}{n}\int\sin^{n-2}(x)dx,
\end{eqnarray}
we obtain
\begin{eqnarray}\label{sin}
&\int_{t}^{\pi/2-t}\sin^{2K}(x)dx\\ \nonumber
&\leq \sin(t)+\left(\frac{2K-1}{2K}\times\cdots\times\frac{1}{2}\right)(\pi/2-2t).
\end{eqnarray}
Similarly,
\begin{eqnarray} \label{cos}
&\int_{t}^{\pi/2-t}\cos^{2K}(x)dx\\ \nonumber
&\leq \sin(t)+\left(\frac{2K-1}{2K}\times\cdots\times\frac{1}{2}\right)(\pi/2-2t).
\end{eqnarray}
Substituting (\ref{sin}) and (\ref{cos}) into (\ref{exp}), gives us 
\begin{eqnarray}\nonumber
&\frac{\mathbb{E}[S]}{N}\leq 2\frac{\sin(t)}{\pi/2-2t}+2\left(\frac{2K-1}{2K}\times\cdots\times\frac{1}{2}\right)\\ \nonumber
&=\frac{2}{\sqrt{300K}(\pi/2-2t)}+2\frac{{2K \choose K}}{4^{K}}\\ \nonumber
&\overset{a}\simeq\frac{2}{\sqrt{300K}(\pi/2-2t')}+\frac{2}{\sqrt{\pi K}}:=\frac{C}{\sqrt{K}}.
\end{eqnarray}
Where (a) uses a well known approximation which is ${2K\choose K}\simeq\frac{4^{K}}{\sqrt{\pi K}}$. Furthermore, we obtain
\begin{eqnarray}\nonumber
&Var[S]=\mathbb{E}[\mathbb{E}[S^{2}|e_{i}, 1\leq i\leq N]]-\mathbb{E}[S]^{2}\\ \nonumber
&=\sum_{i=1}^{N}\mathbb{E}[p(e_{i})]-\mathbb{E}[p(e_{i})]^{2}\\ \nonumber 
&\overset{b}\leq N(C/\sqrt{K}-C^{2}/K).
\end{eqnarray}
(b) is true since $x-x^{2}$ is an increasing function for $0<x<0.5$. Finally, Chernoff's bound will establish the result.
\\

\textit{Proof of Theorem \ref{ltardos}}\\
First we compute the mean and variance of $S$, i.e., the number of rows in $\textbf{A}$ that attackers can not estimate. $S$ can be written as 
$$
S=\sum_{i=1}^{N}\mathbb{I}_{\{\textbf{A}_{i}=[0,...,0]\}}
$$
where $\mathbb{I}$ is the indicator function. Let $e_{i}$ to be the event that $i$th row of $\textbf{A}$ is entirely zeros. This event happens when either $[f_{i,1},...,f_{i,K}]$ is a vector of all zeros or a vector of all ones. Therefore, 
\begin{equation}
p(e_{i})=\prod_{j=1}^{K}p_{j}+\prod_{j=1}^{K}(1-p_{j}),  1\leq i\leq N,
\end{equation}
where $p_{j}=\sin^{2}(r_{j})$ and $r_{j}\sim Uniform(t',\pi/2-t')$.  Now one can obtain
\begin{eqnarray}\nonumber
&\mathbb{E}[||\textbf{f}_{1}-\widehat{\textbf{f}}_{1}||^{2}]\leq \mathbb{E}[S] \\ \nonumber
&=\mathbb{E}[\mathbb{E}[S|e_{i},1\leq i\leq N]]=N\mathbb{E}[p(e_{i})]\\ \nonumber
&\overset{a}=N\left(\mathbb{E}[p_{1}]^{K}+(1-\mathbb{E}[p_{1}])^{K}\right) \\  \label{mean1}
&\overset{b}=N/2^{K-1},
\end{eqnarray}
where (a) is true since $p_{j}$s are i.i.d and (b) is because of $\int_{t}^{\pi/2-t}\sin^{2}(r)\frac{dr}{\pi/2-2t}=1/2$.
\begin{small}
\begin{eqnarray}\nonumber
&\mathbb{E}[S^{2}]=\mathbb{E}[\mathbb{E}[S^{2}|e_{i},1\leq i\leq N]]\\ \nonumber
&\overset{c}=\mathbb{E}[N^{2}p(e_{i})^{2}+Np(e_{i})-Np(e_{i})^{2}]\\ \nonumber
&=(N^{2}-N)\mathbb{E}[p(e_{i})^{2}]+N\mathbb{E}[p(e_{i})]\\ \nonumber
&=(N^{2}-N)\left(\mathbb{E}[p_{1}^{2}]^{K}+\mathbb{E}[(1-p_{1})^{2}]^{K}+2\mathbb{E}[p_{1}(1-p_{1})]^{K}\right)\\ \label{var1}
&+N\mathbb{E}[p(e_{i})]\overset{d}\leq (N^{2}-N)\left(3\left(\frac{3}{8}\right)^{K}\right)+N/2^{K-1} ,
\end{eqnarray}
\end{small}
where (c) is due to the fact that $S|(e_{i},1\leq i\leq N)$ is a binomial distribution with parameter $p(e_{i})$ and (d) is because of $5/16<\int_{t}^{\pi/2-t}\sin^{4}(r)\frac{dr}{\pi/2-2t}<3/8$. Using (\ref{mean1}), (\ref{var1}), and Chernoff's bound, we obtain
\begin{small}
\begin{eqnarray}\nonumber
&P\left(S\geq  N/K\right)<\\ \nonumber
&P\left(S\geq  N/2K+ N/2^{K-1}\right)=P\left(|S-N/2^{K-1}|>N/2K\right)\leq\\ \nonumber
&4K^2\left(3\left(\frac{3}{8}\right)^{K}-4\left(\frac{2}{8}\right)^{K}\right)+\frac{4K^2}{N}\left(2\left(\frac{4}{8}\right)^{K}-3\left(\frac{3}{8}\right)^{K}\right)\\ \nonumber
&\overset{e}\leq 12K^2(\dfrac{3}{8})^K+\dfrac{8K^2}{N}(\dfrac{4}{8})^K,
\end{eqnarray}
\end{small}
(e) is true because of the fact that $3\cdot3^{K}-4\cdot2^{K}<3\cdot3^{K}$ and $2\cdot4^{K}-3\cdot3^{K}<2\cdot4^{K}$. Therefore, for $K>6$ with probability at least $1-12K^2(3/8)^K-8K^2(1/2)^K/N$, we have $S<N/K$. In this case, using the Hoeffding's inequality and union bound we get
\begin{equation}
P\left(\max_{i\in\mathcal{K}}|p_{i}-\widehat{p}_{i}|\geq\rho\right)\leq 2Ke^{-2(N-S)\rho^{2}} ,
\end{equation}
By choosing $\rho=\sqrt{\frac{\log N}{N}}$, we obtain the result.\\

\end{document}